\documentclass[10pt]{article}

\usepackage{graphicx}
\usepackage[utf8]{inputenc}
\usepackage[T1]{fontenc}
\usepackage{hyperref}
\usepackage{url}
\usepackage{booktabs}
\usepackage{amsfonts,amsthm}
\usepackage{nicefrac}
\usepackage{microtype}
\usepackage{natbib}

\newtheorem{theorem}{Theorem}
\newtheorem{lemma}{Lemma}
\newtheorem{corollary}{Corollary}

\title{Secure Convolutional Neural Networks using FHE}
\author{Thomas Shortell \and Ali Shokoufandeh}

\date{Department of Computer Science\\
  Drexel University\\
  Philadelphia, PA 19104\\
  \texttt{tms38@drexel.edu},\texttt{ashokouf@cs.drexel.edu}\\
}

\begin{document}

\maketitle

\begin{abstract}

  In this paper, a secure Convolutional Neural Network classifier is
  proposed using Fully Homomorphic Encryption (FHE).  The secure
  classifier provides a user with the ability to out-source the
  computations to a powerful cloud server and/or setup a server to
  classify inputs without providing the model or revealing source
  data.  To this end, a real number framework is developed over FHE by
  using a fixed point format with binary digits.  This allows for real
  number computations for basic operators like addition, subtraction,
  and multiplication but also to include secure comparisons and max
  functions.  Additionally, a rectified linear unit is designed and
  realized in the framework.  Experimentally, the model was verified
  using a Convolutional Neural Network trained for handwritten digits.
  This encrypted implementation shows accurate results for all
  classification when compared against an unencrypted implementation.
\end{abstract}

\section{Introduction}

Privacy and security are at the forefront and keystone problems of
cloud computing in today's connected world.  Cloud computing allows
users to easily outsource their computations but at the cost of losing
privacy and security.  Maintain the privacy and security of data,
particularly personal information, can be prohibitive on the use of a
cloud computing resource.  Protecting the data and algorithms for
usage on cloud computers opens the aperture to many more applications.
The protection must provide the method to limit an attacker's ability
to use the data.  On a parallel tract, machine learning is an
important technology in use today.  Uses of machine learning
algorithms continue to increase, expand and then become vital
components of today's world.  As a concrete problem to focus on,
consider a situation of needing to outsource the classification of a
Convolutional Neural Network \cite{lecun1995convolutional} and/or
providing a classification from a Convolutional Neural Network to a
user without providing the model nor revealing the original source
data.  In this paper we propose a model for performing  a secure
realization of a Convolutional Neural Network using FHE.




In the FHE framework, a user processes their private data on an
external computationally powerful resource with a level of privacy.
In an abstract view of the process, the user encrypts their data,
securely transfers it to the computing resource, runs the process,
securely transfer back the result, and then decrypts it.
Figure~\ref{fig:cnnfhe} depicts this concept in the case of using a
Convolutional Neural Network.  At the start, the user encrypts their
image and securely transfers the image to the server (gray arrow).
Without loss of generality, image is used to identify the 2D object
that is input to the Convolutional Neural Network.  With the encrypted
image on the server, the convolution layers of the network are run;
followed by the fully connected layer. Despite the simplicity of the
CNN architecture considered in this paper, the proposed FHE model can
be extended to more complex models.  It needs to be noted at this
point that any structure of a CNN can be run.  After the processing is
complete, the user securely transfers the data back for decryption.
Decrypting the result provides the user with the classification
results from the Convolution Neural Network.  While conceptually this
process is very simple, extending the FHE to support the CNN over
reals is a nontrivial task.

\begin{figure}
  \centering
  \includegraphics[scale=0.33]{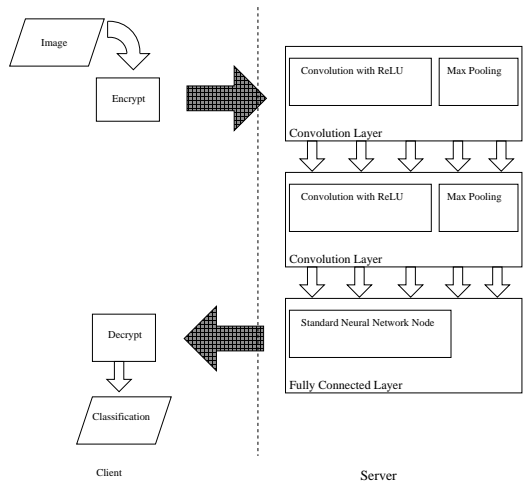}
  \caption{Convolutional Neural Network in a Fully Homomorphic Encryption framework}
  \label{fig:cnnfhe}
\end{figure}

Understanding what a Fully Homomorphic Encryption scheme provides is
critical to understanding the hardness of the problem.  The first
scheme was developed by \citet{gentry2009fully} and which provided the
basic functionality for encrypted data processing using addition and
multiplication.  This research uses a sequel scheme proposed by Gentry
in 2013~\citet{gentry2013homomorphic} which supports two possible
numerical spaces: binary numbers ($\{0,1\}$) and integers
($\mathbb{Z}_q$ for a modulus $q$ which is an FHE parameter).  To
support any processing over real numbers, new models are needed that
represent real numbers in the FHE framework.  As the proposed model
provides an approximation model for original reals, there will be
errors introduced by the representation.  The primary contributions of
this research are the ability to perform a Convolutional Neural
Network while the data is encrypted using an approximation framework
for real numbers, as well as a theoretical bound on the error
introduced by the framework. This in turn enables the user with a
trade off on time/space complexities for accuracy or vice versa.

The basis of the FHE scheme's security hardness is the Learning With
Errors problem by \citet{regev2009lattices}.  The problem hardness is
equivalent to the shortest vector problem in lattices and is
considered quantum hard.  This hardness result is what allows FHE to
process data while encrypted.  In this model, the ciphertexts are
generally represented as matrices and plaintexts are embedded along
the diagonal.  A ciphertext can only be decrypted using the secret key
which is known as the trap door of the hardness problem.  Within this
approach, \citet{gentry2013homomorphic} supports two plaintext spaces:
binary numbers and integers modulo a set value $q$.  The scheme
provides a single operator for binary numbers: \textsc{nand}.  As
\citet{wesselkamper1975sole} showed, \textsc{nand} gates can be used
to create all other binary gates.  In a general set up, the gate can
be computed via matrix-matrix arithmetic $I-C_1*C_2$, which requires
matrix multiplication. This latter computation is the main source of
time and space complexity for the scheme. We will employ the
\textsc{nand} gate to compute fixed point binary numbers in FHE.  Lastly, FHE
schemes have an additional capability to refresh ciphertexts.  This is
will be a critical in our approach because the Learning with Errors
problem introduces noise into the ciphertext, which is sometimes
characterized as error.
We will refer to the Learning with Errors error as noise and the numerical introduced by the framework as error to avoid confusion. 

The Convolutional Neural Networks considered here is based by
LeCun~\citet{lecun1995convolutional} that reduces the number of
trainable parameters in an equivalent fully connected neural network.
This in turn reduces the number of trainable parameters by creating
layers that use a convolution of an image over a smaller subset of the
pixels.  The individual convolutions are known as feature maps and are
equivalent to a dot product of the following form:
\begin{equation}
  \sum_{i=0}^{N} c_i \cdot x_i,
\end{equation}
where $c_i$ is the trained constant, $x_i$ is the input, $N$ is the
entire size of the convolution, and starting at $0$ allows for a bias
value to be introduced.  The computation also assumes that the two
dimensional section is aligned into a single array.  As the
convolutions are done as a standard neural network node, the
additional change is that the activation function is a rectified
linear unit (ReLU):
\begin{equation}
  \label{eq:relu}
  relu(x) = \left\{ 
    \begin{array}{ll}
      x & \mbox{if } x > 0 \\
      0 & \mbox{otherwise} 
    \end{array}
    \right.
\end{equation}
The final step of a convolution layer is subsampling via  max pooling of a set of
outputs.  Max pooling takes in a segment of the image and selects the
max value of the segment.
\begin{equation}
  \label{eq:max}
  \mbox{max}\left( x_1, x_2, \dots, x_n \right)
\end{equation}
A set of convolution layers is used to compute feature maps and reduce
the size to a manageable value of trainable parameters for the fully
connected layer(s).



The remainder of paper is organized as follows.
Section~\ref{sec:related} provides an overview of related work and
current research in the realm of FHE, secure processing, and machine
learning.  Section~\ref{sec:cnninfhe} describes the framework for
realizing Convolutional Neural Networks in FHE.
Section~\ref{sec:experimental} discusses the experimental results of
applying the CNN under FHE framework.  Section~\ref{sec:conclusion}
concludes this paper along with future work.

\section{Related Works}
\label{sec:related}

This section contains details on related works in FHE, secure
processing, and machine learning.  As mentioned in the Introduction,
the first FHE scheme was developed by \citet{gentry2009fully}.  This
original scheme was limited to binary numbers and was inefficient in
comparison to later schemes.  Additional significant work continued in
the early 2010s time frame to improve the FHE schemes from both
functional capabilities and time/space efficiencies (see
\citet{brakerski2012fully,brakerski2012leveled,brakerski2011efficient}).
These improvements enabled processing to include the integer ring vice
binary.  \citet{gentry2013homomorphic} was chosen as the FHE scheme
for this research because it further improved on the earlier schemes
with a simpler conceptual and implementation framework.  Additional
work in FHE continues, particularly with multi-key implementations of
multiple users (see \citet{peikert2016multi,brakerski2016lattice}).
While outside the scope of this paper, multi-key research has
potential uses when considering training of Convolutional Neural
Networks.

FHE is not the only technique for secure processing.
\citet{paillier1999public} implemented an encryption scheme that
allows for additive homomorphic operations.  The Paillier scheme is
not an FHE scheme because Paillier's scheme does not allow for
multiplicative homomorphic operations of two ciphertexts.  Multiple
uses of Paillier have occurred over the years; including for
algorithms such as SIFT \citep{hsu2011homomorphic} and SURF
\citep{bai2014surf}.  While the previous two algorithms are signal
processing related, machine learning algorithms have been researched
as well.  \citet{yu2006privacy} implemented a privacy-preserving SVM
using secure protocol techniques.  \citet{laur2006cryptographically}
also implemented a cryptographically private SVM using secure protocol
techniques around the same time as the previous research.  Another
example of encrypted machine learning classification is
\citet{bost2015machine}, where the authors have implemented a set of
algorithms including SVM, Naive Bayes, Decision Trees, among others;
Convolutional Neural Networks was not included among the others.
\citet{yang2005privacy} is the last example of classification for
Naive Bayes with secure protocol techniques.  On the other side of
machine learning for training the algorithm(s),
\citet{yuan2014privacy} implemented a neural network training using
homomorphic encryption via the BGN scheme.

\citet{chabanne2017privacy} also considered the Convolutional Neural
Networks classification and FHE.  In addition to using the most recent
FHE scheme as the underlying security framework, the implementations
for the rectified linear unit and the max pooling nodes are
implemented using approximations (polynomials and averaging
respectively) whereas this research computes the actual functions
as defined in the unencrypted domain. 
\citet{gilad2016cryptonets} considered CNNs as well using a different FHE scheme known as YASHE \citet{bos2013improved}.
This implementation is similar to \citet{chabanne2017privacy} and has the same similarities and difference with this research.
Two other implementations exist using FHE as part of the techniques but additionally used different neural network implementations with binary \citet{sanyal2017} and discretized neural networks \citet{cryptoeprint:2017:1114}.
Next, two other implementations exist in terms of using oblivious techniques like garbled circuits.
First, \citet{cryptoeprint:2017:452} created an oblivious neural network (called MiniONN) using garbled circuits combined with homomorphic encryption.
Second, \citet{juvekar2017} improved on the previous concept for a lower latency framework using a linear algebra kernel in FHE to perform faster computations.
Finally, \citet{unknown2018} has developed a concept similar to this research; however, our research shows better performance in accuracy and run time. 

Similarly, there is existing FHE
research for using binary numbers to implement reals.  numbers in FHE.
\citet{shortell2017secure} and \citet{shortell2018secure} used binary numbers in fixed point format
to perform an secure FFT.  We improve upon their approach to compute
real functions to perform the Convolutional Neural Network in FHE.

\section{CNN in FHE}
\label{sec:cnninfhe}

This section we present the details of the Convolution Neural Network
in the FHE framework including the numerical error introduced by the
implementation.  Real numbers are stored in the framework using a
fixed point binary format.  As mentioned earlier, a Convolutional
Neural Network is built from two different types of layers:
convolution and fully connected layers.  Figure~\ref{fig:specificcnn}
depicts the Convolutional Neural Network structure that is used in
this research, which is loosely based on the original network used in
\citet{lecun1995convolutional}. Overall, there are 3 layers, starting
with 2 convolutional layers followed by a fully connected layer. 

For the sake of illustration, we assume the input to the network is
the handwritten decimal digits in the form of a $28\times 28$ image,
with 10 output nodes. The first convolution layer reduces the size to
$24\times 24$ using a $5\times 5$ convolution across the image.  Max
pooling operates on $2\times 2$ segments without any overlap causing
the feature map to reduce to $12\times 12$ in size.  In the first
layer, there are 4 different convolutions that are run; thus the 4
feature maps outputs of the layer.  In the second layer, the number of
convolutions is 15 with interactions being allowed between the feature
maps of the first layer.  The second layer uses the same control
parameters as the first layer: $5\times 5$ convolution and $2\times 2$
max pooling segments.  This leads to output feature map size of
$4\times 4$ from the layer.  With the completion of the two
convolution layers, the fully connected layer takes the $240$ ($15
\times 4 \times 4$) outputs of the convolution layer and inputs them
directly into 10 output nodes of the fully connected layer.  In terms
of FHE, the major components that need to be developed are the ability
to do a dot product, the ReLU, and max function.

\begin{figure}
  \centering
  \includegraphics[scale=0.33]{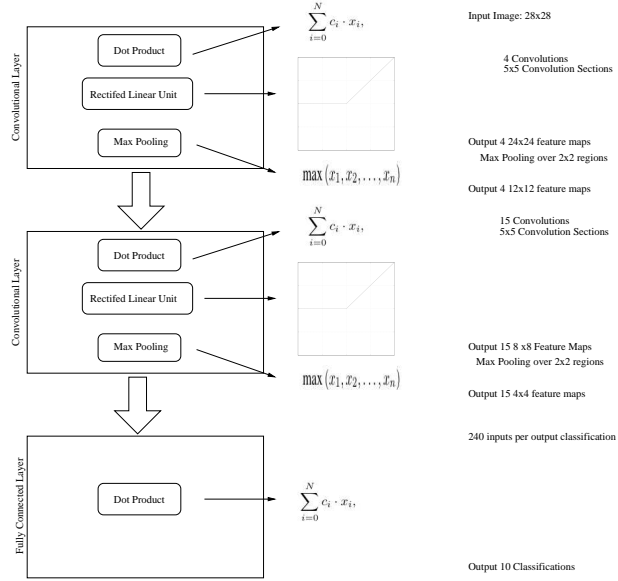}
  \caption{Specific Convolutional Neural Network Used}
  \label{fig:specificcnn}
\end{figure}

\paragraph{Fixed Point Binary Format}
We will use the fixed point binary representation to support real-valued
computations for FHE.  The format is defined by multiplying a real
number, $r$, by a integer scaling factor, $\delta$, and then flooring
the result.  This provides an integer, $z$:
\begin{equation}
  \label{eq:scale}
  \lfloor r \cdot \delta \rfloor = z.
\end{equation}
This number can be converted to binary via:
\begin{equation}
  z = \sum_{i=0}^n b_i \cdot 2^i.
\end{equation}
Given the number in the binary format, fixed point operations in the
integer space are implemented.  Basic operations including addition,
subtraction, multiplication, and division can be implemented using
binary gates starting from the \textsc{nand} gate to other binary
gates to half and full adders. While addition and subtraction are
straight forward to implement with full adders, computing generic
multiplication requires the use of Wallace Trees
\citet{wallace1964suggestion}. We note that division can be computed
as well but this operation is not necessary for the implementation of
Convolutional Neural Network.  As two's compliment is being used in
the binary format, secure comparisons are done by examining the bits.
By subtracting two numbers, the result (in encrypted form as a single
bit) can identify equality, less than, greater than, or lack thereof.
The single bit has many uses including obfuscation.  An example
consider
\begin{equation}
  sum = sum + b \cdot value.
\end{equation}
that shows conditional operations can be done by using the resultant
bit ($b$ in this case) to add $value$ to the $sum$.  We will use this
simple but powerful technique to implement the rectified linear unit and the max
pooling capabilities.

\paragraph{Encrypted Dot Product}
Computing a dot product within the framework is crucial to computing
the Convolutional Neural Network.  A dot product can be defined as:
\begin{equation}
  \sum_{i=0}^m c_i \cdot x_i.
\end{equation}
As the fixed point binary format provides addition and
multiplication, computing the dot product is straight forward.  An
important final note for Convolutional Neural Networks is that
constant multiplication can be used in the case of client
classification which the client does not need to know the parameters,
i.e. the server is not required to encrypt the parameters.

\paragraph{Activation Functions - ReLU}
Rectified Linear Unit (ReLU) is required to properly implement the
Convolutional Neural Network (see Eq.~\ref{eq:relu}).  First we need
to determining if the input is greater than or equal to
zero.  If this result is called $b$, then the ReLU is computed as:
\begin{equation}
  res = b \cdot x .
\end{equation}
This technique accurately computes ReLU. 

\paragraph{Max Pooling}
Implementing max pooling in FHE is an operation where the max value of
the set of inputs is determined; i.e. the max function
(Eq.~\ref{eq:max}).  Because a secure comparison can be used to
identify the greater than of two ciphertexts, the max function is
implementable by running a loop over the inputs and finding the max
value.  Using obfuscation, the update equation can be stated as:
\begin{equation}
  max = b \cdot currentmax + !b \cdot next
\end{equation}
where $b$ and $!b$ identify which of the two values is the max
(i.e. between the $currentmax$ and $next$).

Next, we estimate the magnitude of errors introduced during these
computations. The main source of the error is due to the use of floor
function in Eq.~\ref{eq:scale}, which increase over iterative
computations.  We first state the theorem describing the error bound:

\begin{theorem}
  \label{thm:cnnerror}
  Given the FHE implementation of a Convolutional Neural Network with a set of convolutions layers and fully connected layers $L$, the implementation will result in an error of:
  \begin{equation}
    \Delta \cdot \prod_{i \in L} r_i \cdot d_i 
  \end{equation}
  where $\Delta$ is the initial error, $r_i$ and $d_i$ are defined per
  each layer as follows: for a convolution layer, $r_i$ is square root
  of convolution size and $d_i$ is the maximum of the sum of constants
  for the set of convolutions for the layer.  for a fully connected
  layer, $r_i$ is the square root of the node input size and $d_i$ is
  the maximum of the sum of constants among the outputs of the layer.
\end{theorem}

Individual lemmas are built up from smaller pieces to larger parts of
the Convolutional Neural Network.  First we estimate the error bound
on the ReLU processing and follow that the max pooling processing.
This will be followed by the error bound on the dot product,
convolutional, and fully connected layers, respectively.


\begin{lemma}
  \label{lem:relu}
  Given an FHE based ReLU implementation, the computation will not
  increase the error of the input.
\end{lemma}

\begin{proof}
  As the ReLU computation only outputs the same value or zero, the
  only error is the one at the start of the process.  Thus the ReLU
  implementation does not introduce any error.
\end{proof}

\begin{lemma}
  \label{lem:maxpool}
  Given an FHE based max pooling implementation, the computation
  will not increase the error of the input.
\end{lemma}

\begin{proof}
  As max pooling takes in a set of values and outputs the max, the
  only error is the error at the start of the process as this will be
  continued forward.  Thus the max pooling implementation does not
  introduce any error.
\end{proof}


\begin{lemma}
  \label{lem:dotproduct}
  Given an FHE based dot product implementation, the computation will
  increase the error of the output value by:
  \begin{equation}
    \Delta \cdot \sqrt{s} \cdot |c|
  \end{equation}
  where $\Delta$ is the max introduced error from the inputs, $s$ is
  the size of the vector, and $|c|$ is the magnitude of the
  convolution constants.
\end{lemma}

\begin{proof}
  As the constant will not have any error, the focus is on the input
  parameters.  Using,
  \begin{equation}
    \sqrt{ \left(\frac{\partial F}{\partial x_1} \delta x_1 \right)^2 + \cdots + \left(\frac{\partial F}{\partial x_n} \delta x_n \right)^2 }
  \end{equation}
  to propagate the error, the individual partial derivatives will be the individual constants.
  The $\delta x_i$ values will be $\Delta$ as this is the input error.
  This reduces the equation to:
  \begin{equation}
    \sqrt{ \sum\left(c_i^2 \cdot \Delta^2 \right) }.
  \end{equation}
  Extracting $\Delta$ from the sum, the remainder is the magnitude of
  the constants multiplied by the size of the vector.  Thus the error
  is:
  \begin{equation}
    \Delta \cdot \sqrt{s} \cdot |c|.
  \end{equation}
\end{proof}

\begin{lemma}
  \label{lem:convolution}
  Given the FHE implementation of the convolution layer, the
 computation will increase the error by
  \begin{equation}
    \Delta \cdot \sqrt{s} \cdot |c|
  \end{equation}
  where $\Delta$ is the max input error per input value, $\sqrt{s}$ is
  the square of the convolution size, and $|c|$ is the magnitude of
  the convolution constants.
\end{lemma}

\begin{proof}
  To prove this lemma, it is necessary to use
  Lemmas~\ref{lem:relu},~\ref{lem:maxpool}, and~\ref{lem:dotproduct}.
  The first step of the convolution layer is the dot product, which
  will introduce an error of $\Delta \cdot |c|$.  This was proved by
  Lemma~\ref{lem:dotproduct}.  The second step is the activation
  function post the dot product which is an ReLU, which as
  Lemma~\ref{lem:relu} proved introduces no additional error; so error
  is $\Delta \cdot \sqrt{s} \cdot |c|$.  The final step is the max
  pooling part.  Lemma~\ref{lem:maxpool} proved that no additional
  error is introduced.  Thus, the error introduced by a single
  convolution layer is $\Delta \cdot \sqrt{s} \cdot |c|$.
\end{proof}


\begin{lemma}
  \label{lem:fullyconnected}
  Given the FHE implementation of the fully connected layer, the
  implementation will increase the error by
  \begin{equation}
    \Delta \cdot \sqrt{s} \cdot |f|
  \end{equation}
  where $\Delta$ is the max input error per input value, $\sqrt{s}$ is
  the square of the inputs to the layer, and $|f|$ is the magnitude of
  the fully connected layer constants.
\end{lemma}

\begin{proof}
  In this context, the fully connected layer is a dot product with an
  activation function that is linear.  With $\Delta$ as the input
  error, Lemma~\ref{lem:dotproduct} proves this error will be:
  \begin{equation}
    \Delta \cdot \sqrt{s} \cdot |f|    
  \end{equation}
  where $|f|$ has been substituted for $|c|$.  As a linear function is
  simply an identity function which will not introduce any error
  beyond the error already included.
\end{proof}


Theorem~\ref{thm:cnnerror} can be proved using the results of
Lemma~\ref{lem:convolution} and Lemma~\ref{lem:fullyconnected}.

\begin{proof}
  Given the structure of the Convolutional Neural Network as $L$, each
  layer is input to the next layer.  This means the error will
  propagate over time from layer to layer.
  Lemmas~\ref{lem:convolution} and~\ref{lem:fullyconnected} provide
  the error contributed by a single layer for convolutions and fully
  connected respectively.  As the error is propagated from layer to
  layer, the maximum error is what should be retained to the next
  layer.  For the convolutional layer (Lemma~\ref{lem:convolution}),
  this will be the value $|c|$.  For the fully connected layer
  (Lemma~\ref{lem:fullyconnected}), this will be the value $|f|$.  The
  initial error, $\Delta$, can be pulled out at each layer.  By
  designating the individual $|c|$ and $|f|$ as $d_i$ for the
  individual layers and $\sqrt{s}$ as $r_i$ per layer, the error is
  increased from the original error to the final error by the product
  of each individual layer.  This equates to the following equation:
  \begin{equation}
    \Delta \cdot \prod_{i \in L} r_i \cdot d_i 
  \end{equation}
\end{proof}


Consider the CNN in Figure~\ref{fig:specificcnn}, the following
corollary shows the specific situation for that structure as an example of applying Theorem~\ref{thm:cnnerror} to a specific situation:

\begin{corollary}
  \label{cor:current}
  For a CNN with 2 convolutional layers and a fully connected layer,
  the total error will be:
  \begin{equation}
    \Delta \cdot 25 \cdot \sqrt{240} \cdot d_1 \cdot d_2 \cdot d_3
  \end{equation}
  where the $r_i$ values have been filled in.  $d_i$ values are
  dependent on the specific training.
\end{corollary}

\section{Experimental Results}
\label{sec:experimental}

To verify Theorem~\ref{thm:cnnerror}, an implementation of the CNN is
necessary.  The code was developed in C++ using a home grown FHE
implementation supported by FFLAS-FFPACK (\citet{fflasffpack}) and
HDF5 (\citet{hdf5}).  Input values are real numbers between -1 and 1.
Fixed point binary format is configured to be 32 total bits with 16
bits being the fractional space; requiring a fixed point scale of
65536.  The actual CNN used is based on the original CNN paper
\cite{lecun1995convolutional}.  Input images are handwritten numerical
digits from 0 to 9.  The same data set as
\cite{lecun1995convolutional} was used.  Training was done offline
from the FHE implementation; training was not part of the FHE
implementation.  Tiny CNN (\cite{tinycnn}) was used to train the
network using the original training set.  Individual images from the
test set were run against the implementation.  The ten outputs between
unencrypted and encrypted implementation are compared to verify the
encrypted Convolution Neural Network performs as well as the
unencrypted implementation.  These results are discussed in the next
paragraph.

Using a set 20 images from the original test set, each image was
encrypted, run through the Convolutional Neural Network implementation
and the classifications decrypted for analysis.  All encrypted results
matched the unencrypted recommended classifications exactly.
Numerical errors where contained within the $10^{-3}$ region.  On
average, the individual error was 0.00085 with a standard deviation of
0.00022.  These results are aligned with Theorem~\ref{thm:cnnerror}
and Corollary~\ref{cor:current}; clearly showing that the
Convolutional Neural Network implementation successfully works in the
encrypted domain.


Time and space complexity are an important consideration in FHE.
FHE's basic time and space complexity have been considered
inefficient.  A single FHE ciphertext is a matrix $O(\tilde{n}^2)$;
this is a single bit in the fixed point binary format.  $\tilde{n}^2$
is the size of the FHE ciphertext matrix and worse case operations are
$O( \tilde{n}^3 )$.  All algorithm complexity is multiplied by this.

From a space complexity perspective, there will be interim space
needed during the computations and output space for the result.
Starting with the interim space, a convolution layer has $l$ feature
maps with $m^2$ as the individual feature map size.  Combining, this
is bounded by $O(l \cdot m^2 \cdot f \cdot \tilde{n}^2)$ where $f$ is
used for the number of fixed point binary digits.  As this is for a
single layer, each layer will have different parameters.  Considering
$l$ and $m$ as the max values from all the layers, the asymptotic
bound for the interim space complexity can include the number of
layers, $n_l$, thus $O( n_l \cdot l \cdot m^2 \cdot f \cdot
\tilde{n}^2)$.  Looking at the output space complexity, the fully
connected layer will output a set number of classifications,
designated by $c$.  As these are individual ciphertexts of the format,
the asymptotic bound will be the combination: $O(c \cdot f \cdot
\tilde{n}^2)$.

We note that the \textsc{nand} gate that is the basis of the FHE
operations is $O(\tilde{n}^3)$.  For a single convolutional layer,
there are $l$ feature maps of size $m^2$ that each have a convolution
of a $5 \times 5$ size.  The convolution is effectively 25
multiplications and additions (designated $25am$).  So far, this time
complexity is $O( l \cdot m^2 \cdot (25am))$.  Next, the $25am$ needs
to be converted into individual \textsc{nand} gates.  Multiplications
in the fixed point format start with a $4f$ operation (\textsc{and}
gates of all digits) and then a $\log f$ operation of $2f$ full
adders.  Full adders need a total of $36$ \textsc{nand} gates.
Combining these terms becomes $O( f^2 \log{f} )$.  Addition takes $f$
full adders; thus $36f$ \textsc{nand} gates.  Thus a single layer's
asymptotic time complexity is $O( l \cdot m^2 \cdot f^2 \log{f} \cdot
\tilde{n}^3)$.  A fully connected layer is similar by changing to the
number of classifications and the proper inputs.  Becoming O$( c \cdot
m^2 \cdot f^2 \log{f} \cdot \tilde{n}^3)$.  Considering a finite
number of convolution layers, $n_l$, and fully connected layers,
$n_c$, the asymptotic time complexity is $O( n_l \cdot l \cdot m^2
\cdot f^2 \log{f} \cdot \tilde{n}^3 + n_c \cdot c \cdot m^2 \cdot f^2
\log{f} \cdot \tilde{n}^3)$.  This asymptotic time complexity is worst
case because after the first layer the number $m^2$ drops.

It is worth noting that parallel processing techniques are available
to improve actual running time.  The current implementation used a
GPU to increase efficiency of the matrix-matrix multiplication.
Additionally, the individual convolutions in a layer do not depend on
each other, thus these can be run in parallel.  For concrete numbers,
the implementation was able to compute the classifications in two days
using a four CPU desktop computer with two GPUs and total space of
about one GB.  As the desktop computer used is not a real cloud
server, greater efficiency could be gained by using more CPUs and
GPUs to bring the two day computation down to a user tolerable time.

\section{Conclusion}
\label{sec:conclusion}


As has been shown in the preceding sections, Convolutional Neural
Network classifications can be computed over FHE.  This enables a user
to offload their computations to a cloud computer or for a
client/server model of protecting privacy of the individual data.  An
open problem to be considered is potentially using the encrypted
classification for other algorithms post classification.  An argmax
function is easy to implement in FHE given the work done in
Section~\ref{sec:cnninfhe} (basically an extension of the max
function).  Argmax allows for selecting the classification for other
uses.



\bibliographystyle{plainnat}
\bibliography{mlconfpaper}

\end{document}